\documentclass{llncs}
\usepackage{graphicx}
\usepackage{epsfig}
\usepackage{amsmath,amsfonts,amssymb}
\newtheorem{obs}{Observation}
\newtheorem{col}{Corollary}

\usepackage{wrapfig}
\usepackage{subfig}
\usepackage{color}

\newcommand{\remove}[1]{}

\begin{document}
\title{Color Spanning Annulus: Square, Rectangle and Equilateral Triangle}
\author{Ankush Acharyya, Subhas C. Nandy and Sasanka Roy}
\institute{Indian Statistical Institute, Kolkata, India}

\maketitle

\begin{abstract} In this paper, we study different variations of minimum width 
color-spanning annulus problem among a set of points $P=\{p_1,p_2,\ldots,p_n\}$ 
in $I\!\!R^2$, where each point is assigned with a color in $\{1, 2, \ldots, k\}$. 
We present algorithms for finding a minimum width color-spanning axis parallel 
square annulus $(CSSA)$, minimum width color spanning axis parallel rectangular
annulus $(CSRA)$, and minimum width color-spanning equilateral triangular annulus
of fixed orientation $(CSETA)$. The time complexities of computing (i) a 
$CSSA$ is $O(n^3+n^2k\log k)$ which is an improvement by a factor $n$ over the
existing result on this problem, (ii) that for a $CSRA$ is $O(n^4\log n)$, and for (iii) 
a $CSETA$ is $O(n^3k)$.  The space complexity of all the algorithms is $O(k)$.
\end{abstract}


\section{Introduction}
The motivation for studying color-spanning 
objects stems from the facility location problems, where we may have different 
types of facilities, and the objective is to identify a location of desired 
shape with at least one copy of each facility and the measure of the region is 
optimized. In this paper, we study the minimum width color-spanning annulus 
problem for different objects.

A point set $P=\{p_1,\ldots,p_n\}$ is given in $I\!\!R^2$. Each point 
$p\in P$ is assigned a color from the set $\chi=\{1,2,\ldots, k\}$ of $k$ 
distinct colors. There exists at least one point of each color. A region is 
{\em color-spanning}  if it contains at least one point of each colors. 
An annulus $\cal A$ is a region bounded by two co-centric homothetic closed curves 
$C_{in}$ (inner curve) and $C_{out}$ (outer curve).
The (common) center $c$ of $C_{in}$ and $C_{out}$ is 
referred to as the {\em annulus-center}, and the {\em width} of the 
annulus is the Euclidean distance between two closest points on the boundary of 
$C_{in}$ and $C_{out}$ respectively. 
In this paper, we are interested to find square, rectangular and triangular annulus The objective is to minimize the width of the annulus.

\textcolor{blue}{\it Related Work:}
The minimum width annulus problem is well studied in the literature. The most 
common variation of this problem is the color-spanning circle. Abellanas et al. 
\cite{abellanas2001smallest} showed that the smallest color spanning circle can 
be computed in $O(kn\log n)$ time. Abellanas et al. \cite{abellanas2001farthest} 
also showed that the narrowest color-spanning strip and smallest axis-parallel 
color-spanning rectangle can be found in $O(n^2\alpha(k)\log k)$ and 
$O(n(n-k)\log^2n)$ time respectively.  Das et al. \cite{das2009smallest} 
improved the time complexity of narrowest color spanning strip problem to 
$O(n^2\log n)$, and smallest color-spanning axis-parallel rectangle problem to 
$O(n(n-k)\log k)$. They also provided a solution for the arbitrary oriented 
color-spanning rectangle problem in $O(n^3\log k)$ time using $O(n)$ space. 
Recently, Khanteimouri et al. \cite{khanteimouri2013computing} presented an 
algorithm for color spanning square in $O(n\log^2 n)$ time. Khanteimouri et al. 
\cite{hasheminejadcomputing} also presented a solution for the color spanning 
axis-parallel equilateral triangle in $O(n\log n)$ time.

On the other hand, the problem of computing the minimum width annulus is also 
studied in the literature. Given a set of $n$ points, computing the minimum 
width circular annulus containing all the points was independently addressed in 
\cite{de2000computational,roy1992establishment,wainstein1986non}.  All of their 
methods result in a time complexity of $O(n^2)$. There are also sub-quadratic 
time algorithms for the circular annulus problem. Using parametric searching 
technique, Agarwal et al. \cite{agarwal1992applications} presented an 
$O(n^{\frac{8}{5}+\epsilon})$ time algorithm. Agarwal et al. 
\cite{agarwal1996efficient} also presented a randomized algorithm for the same 
problem which runs in $O(n^{\frac{3}{2}+\epsilon})$ time. To talk about the 
variations other than circular annulus for a general point set in $I\!\!R^2$, 
the well known results are an $O(n)$ time optimum algorithm for the axis 
parallel rectangular annulus by Abellanas et al.  \cite{abellanas2004best}, 
and an $O(n\log n)$ time optimum algorithm for the axis parallel square annulus 
by Gluchshenko et al. \cite{gluchshenko2009optimal}. Mukherjee et al. 
\cite{mukherjee2011minimum} proposed an algorithm for computing the minimum 
width axis parallel rectangular annulus for a point set in $I\!\!R^d$ in $O(nd)$ 
time. They also proposed an algorithm for arbitrary oriented minimum width 
rectangular annulus in $I\!\!R^2$ that runs in $O(n^2\log n)$ time using $O(n)$ 
space. Recently, Bae \cite{bae2016computing} proposed the minimum width square 
annulus of arbitrary orientation in $O(n^3\log n)$ time.  The color spanning 
annulus problem is comparatively new in the literature. Acharyya et al. 
\cite{acharyya2016minimum} presented two algorithms for finding the minimum 
width color spanning circular and axis-parallel square annulus. Both the 
algorithms run in $O(n^4\log n)$ time using $O(n)$ space.

\textcolor{blue}{\it Main Contribution:} 
In this paper, we propose algorithms to compute the minimum width color-spanning 
annulus where $C_{in}$ and $C_{out}$ are (i) axis-parallel squares $(CSSA)$, 
(ii) axis-parallel rectangles $(CSRA)$, and (iii) equilateral triangles of fixed 
orientation $(CSETA)$. The time complexities of the proposed algorithms are: (i)  for a $CSSA$ is 
$O(n^3+ n^2k\log k)$, (ii) for a $CSRA$ is $O(n^4\log n)$, and (iii) for a $CSETA$ 
is $O(n^3k)$. The space complexity of all the algorithms is $O(k)$. The algorithm 
for $CSSA$ is an improvement of the existing result of \cite{acharyya2016minimum} on 
this problem by a factor of $n$. Moreover, if $k$ is constant, then the improvement factor is
$n\log n$.

\section{Preliminaries}

Interior of an annulus $\cal {A}$, defined by $INT({\cal A})$, is the 
region inside $\cal {A}$ excluding $C_{in}$ and $C_{out}$. 
\begin{obs}\cite{acharyya2016minimum} \label{distinct-color}
 The points of distinct color lying on $C_{in}$ and $C_{out}$ are said to define an annulus $\cal A$, 
and $INT(\cal A)$ does not contain any point of color same as those defining the
annulus $\cal A$.
\end{obs}

Throughout the paper, we use $h_p$ and $v_p$ to denote the horizontal and vertical lines 
passing through point $p$. $d(p_i,p_j),d_\infty(p_i,p_j)$ denote the distance 
between the pair of points $p_i,p_j \in P$ in $L_2$ and $L_\infty$ norm respectively. The closest 
distance of a line segment $\ell$ from a point $p_i\in P$ will be denoted by 
$d(\ell,p_i)$. We also denote $color(p)$ as the color of a point $p\in P$.

\section{Axis Parallel Color-spanning Square Annulus} 


An {\em axis-parallel color-spanning square annulus $(CSSA)$} is a color-spanning annulus $\cal A$
bounded by two co-centric axis-parallel squares $C_{out}$ and $C_{in}$. In \cite{acharyya2016minimum},
it is shown that either $C_{in}$ or $C_{out}$ of a minimum width square annulus $(CSSA)$ has two points
of different colors in its two boundaries. We prove a stronger claim.

\begin{lemma}
Either $C_{in}$ or $C_{out}$ of a $CSSA$ has two points of distinct color on its two mutually parallel boundaries.
\end{lemma}

\begin{proof}
For a contradiction, let only the two mutually perpendicular boundaries (say top and left boundaries) 
of the $C_{out}$ of a $CSSA$ contain two points of different colors. We can reduce the size of $C_{out}$ 
(as well as $C_{in}$) by moving the bottom boundary upward and the right boundary to the 
left of both $C_{out}$ and $C_{in}$ with the same speed until the bottom boundary or the 
right boundary of $C_{out}$ touches a point. Observe that, during this movement if the 
bottom and/or right boundaries of $C_{in}$ encounter some point, those points will enter 
in the annular region, but no point goes out from the inside to the outside of the annular
region. Thus, the annulus remains color-spanning and its width does not increase.

If two mutually perpendicular boundaries of $C_{in}$ contain two points of 
different colors, then also both $C_{in}$ and $C_{out}$ can be expanded keeping 
the annular region color-spanning such that the width of the annulus does not
increase until the lemma is satisfied by $C_{in}$.   Thus the lemma follows. \qed 
\end{proof}


We consider each pair of bi-colored points $p,q\in P$ to define 
the mutually parallel boundaries of $C_{out}$ and compute the minimum width annulus. 
Similar method works for defining the mutually parallel boundaries of $C_{in}$ with 
bi-colored pair of points in $P$.

\begin{lemma} \label{outer-inner}
If $C_{out}$ of a minimum width axis-parallel square annulus $\cal A$, is defined 
by two points $p,q \in P$ on its two parallel boundaries, and its 
annulus-center is fixed at a point $c$, then its $C_{in}$ must pass through a point $r$ with 
minimum $L_\infty$ distance among the farthest points of every color $i \in \{1,2,\ldots k\} 
\setminus \{color(p),color(q)\}$ from the annulus-center $c$.
\end{lemma}
\begin{proof}
 As the annulus-center is fixed at the point $c$ and its radius is also 
 fixed $d_\infty(c,p)$, $C_{out}$ is fixed. Now, 
 for each color $i \in \{1,2,\ldots k\} \setminus \{color(p),color(q)\}$, the point 
 of color $i$ having farthest distance must lie inside the annular region. Since 
 all these $k-2$ points need to be included in the annular region, $C_{in}$ will be 
 defined by one such point which is closest to $c$. \qed
\end{proof}

Let $y(p)>y(q)$, and consider the horizontal strip of width $\delta=y(p)-y(q)$, defined by 
the horizontal lines $h_p$ and $h_q$. If $|x(p)-x(q)| >\delta$, then $p,q$ can not 
define $C_{out}$. Otherwise the points $p,q$ define $C_{out}$ whose center lies on a horizontal 
interval $C=[a,b]$, where $a=(\min(x(p),x(q))+\frac{\delta}{2},\frac{y(p)+y(q)}{2})$ 
and $b=(\max(x(p),x(q))-\frac{\delta}{2},\frac{y(p)+y(q)}{2})$. This configuration 
always holds as shown in \cite{bae2016computing,gluchshenko2009optimal}.
\begin{figure}[h]
 \centerline{\includegraphics[scale=0.5]{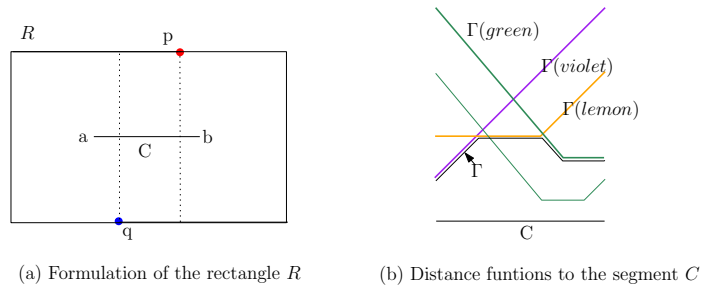}}
 \caption{Illustration of Square Annulus construction}
  \label{sq_algo}
\end{figure}

Consider the left boundary of the square centered at $a$ and the right boundary of the 
square centered at $b$, both of radius $\frac{\delta}{2}$. These boundaries along with 
the horizontal lines through $p,q$ defines a rectangle $R$ (see Fig. \ref{sq_algo}(a)). 
All possible feasible $C_{out}$ are contained in $R$. We consider the subset $P'\subseteq P$ 
within $R$ and verify whether $P'$ is color spanning by using a linear scan. If $P'$ is 
not color spanning, then we can not have any $CSSA$ contained in $R$. Hence 
we discard the horizontal strip defined by the pair $p,q$. Otherwise, for each point 
$r \in P'$ ($color(r) \not\in \{color(p),color(q)\}$), we plot its distance 
in $L_\infty$ metric from different points of line segment $C$. Each of these distance 
functions $f(r)$ is a combination of line segments with slopes in $\{1,0,-1\}$ as described 
in \cite{acharyya2016minimum,bae2016computing} (see Fig. \ref{sq_algo}(b)). 

\begin{lemma} \cite{acharyya2016minimum,bae2016computing} \label{jjj}
The distance curves $f(r)$ and $f(s)$ for the points $r,s \in P \cap R$ respectively 
intersect at exactly one point.
\end{lemma}

We consider the functions $F_i$ for all points each color $i\in \{1,2,\ldots,k-2\}$ separately. Let 
$\Gamma(i)$ be the upper envelope of the functions of $F_i$, $i=\{1,2,\ldots,k\}$. 
For each point $\alpha \in C$, if the vertical line drawn at $\alpha$ intersects 
$\Gamma(i)$ at a point $\beta$, and $\beta$ lies on $f(r) \in F_i$, then the point 
$r$ of $color(i)$ is closest to $C_{out}$ centered at $\alpha$ among all points 
of $P$ having $color(i)$ inside $R$. Observe that, each $\Gamma(i)$ is one of the 
forms listed in Figure \ref{gamma}, and their corresponding symmetric forms. This 
follows from the fact that the upper envelope $\Gamma(i)$ of the curves in $F(i)$
consists of at most one line segments of slopes $\{-1,0,1\}$. 

\begin{figure}[h]
 \centerline{\includegraphics[scale=0.5]{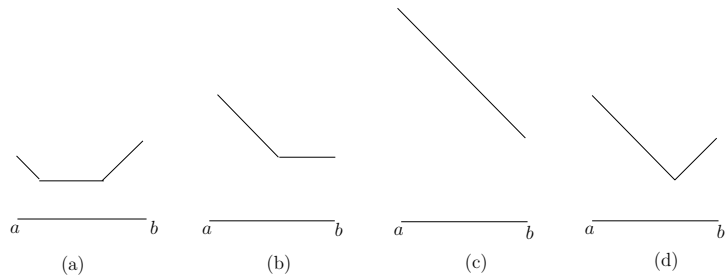}} 
 \caption{Nature of $\Gamma(i)$}
  \label{gamma}
\end{figure}

Now, to compute the furthest points of each color $i$ from $C$, we need to consider the lower envelope 
$\Gamma$ of $\Gamma(i), i \in \{1,2,\ldots,k\} \setminus \{color(p),color(q)\}$.
Again to minimize the width of the annulus, we choose the point on $\Gamma$ having 
maximum vertical distance from $C$. Its projection $c$ on $C$ is the center of the 
$CSSA$ with $p,q$ on the two parallel sides of its $C_{out}$ (see Fig. \ref{sq_algo}(b)).

\begin{lemma}\label{envelope}
Given the rectangle $R$, the set of points within $R$, we 
can find the annulus-center $c$ in $O(n\log n)$ time.
\end{lemma}
 \begin{proof}
  Let $color(p)=k$, $color(q)=k-1$, and there are $n_i$ points of color $i=\{1,2,
  \ldots,k-2\}$. By Lemma \ref{jjj}, the computation of upper envelope $\Gamma(i)$ 
  takes $O(n_i\log n_i)$ time \cite{sharir1995davenport}. Thus, the total time 
  for computing $\Gamma(i)$ for all $i=1,2,\ldots, k-2$ is $O(n\log n)$. Now, we have 
  $k$ totally defined functions $\Gamma(i)$ in the interval domain $C=[a,b]$, where 
  each pair of functions $\Gamma(i)$ and $\Gamma(j)$ intersect in at most two points. 
  The size of the lower envelope $\Gamma$ of the functions $\Gamma(i), i=1,2,\ldots,
  k-2$ is $\lambda_2(k-2)$ (the Davenport Schinzel sequence of order 2), and it can 
  be computed in $O(\lambda_2(k-2)\log k)$ time \cite{sharir1995davenport}. Finally, 
  we compute the point having maximum $y$-coordinate on $\Gamma$ by inspecting all 
  its vertices. Since $\lambda_2(k-2)=2k-5$, the total time required for processing 
  $\Gamma(i), i=1,2,\ldots,k-2$ is $O(k\log k)$. Thus, the total time complexity is 
  dominated by computing $\Gamma(i), i=1,2,\ldots,k-2$, which is $O(n\log n)$. \qed 
 \end{proof}

 \begin{lemma}
Given a set of $n$ points, each assigned with one of the $k$ given colors, 
the minimum width $CSSA$ can be computed in $O(n^3\log n)$ time.

\end{lemma}
\begin{proof}
 We consider $O(n^2)$ pairs of bi-colored points. For each pair of such points 
 $(p,q)$, we can construct the rectangle $R$ in $O(1)$ time. With
 the help of a linear search we can verify whether $R$ is color-spanning. If $R$ is color-spanning, then using 
 lemma \ref{envelope}, we can determine the optimum square
 annulus contained in $R$ in $O(n\log n)$ time.
 Thus the time complexity result follows. \qed 
\end{proof}

We can further improve the time complexity in the following way using a total $O(k)$ amount of extra space.
We choose a pair of points $p,q \in P$ and let $|x(p)-x(q)| >\delta$, where $\delta=y(p)-y(q)$.
Now, we can construct the rectangle $R$ and verify its color-spanning property as was done earlier. If it is color spanning,
then for each color we maintain the upper envelope $\Gamma_i$. As mentioned earlier, $\Gamma_i$'s
are of constant complexity. Thus, each $\Gamma_i$ can be maintained using $O(1)$ space. For each point 
$r\in P$, if $r$ is inside $R$, then we consider its distance curve $f(r)$ from $C$, and update $\Gamma_i$ for 
$i=color(r)$ considering the intersection of $f(r)$ and the existing $\Gamma_i$. This can be done in $O(1)$ time. Thus considering all points
in $P$, we can construct the $\Gamma_i, i=\{1,2,\ldots,k-2\}$ in $O(n)$ time using $O(k)$ space. At the end, we consider
the lower envelope $\Gamma$ of these $\Gamma_i$'s and return the maximum point. As in Lemma \ref{envelope}, this can be done in $O(k\log k)$
time using $O(k)$ space. Considering $O(n^2)$ pairs of bi-colored points, we have the following result:
\begin{theorem}
 Given a set of $n$ points, each assigned with one of the $k$ given colors, 
the minimum width $CSSA$ can be computed in $O(n^3+n^2k\log k)$ time using $O(k)$ extra space. 
\end{theorem}

%

\section{Axis Parallel Color-spanning Rectangular Annulus} 
An {\em axis-parallel color-spanning rectangular annulus $(CSRA)$} is a 
color-spanning annulus $\cal A$ bounded by two co-centric axis-parallel 
rectangles $C_{out}$ and $C_{in}$. The top (resp. bottom, left, right) boundaries
of $C_{in}$ and $C_{out}$ are said to be {\it similar sides} of these two rectangles.
The width of $CSRA$ is half of the difference of lengths (widths) of $C_{out}$ 
and $C_{in}$ (see Fig. \ref{fig:axis_parallel_rect}).  
\begin{figure}[h]
 \centerline{\includegraphics[scale=0.4]{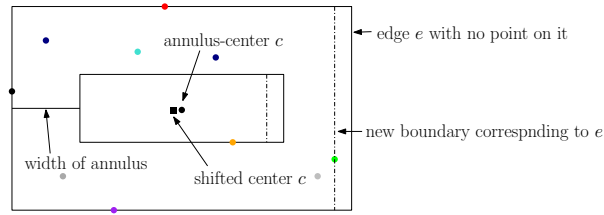}}
 
 \caption{Rectangular Annulus} 
  \label{fig:axis_parallel_rect}

\end{figure}
\begin{lemma}\label{four_points}
The necessary and sufficient condition for $\cal A$ to be a minimum width 
$CSRA$ is that (i) all the four edges 
of $C_{out}$ must contain at least one point, and at least one edge of $C_{in}$
must contain at least one point, or (ii) all the four edges 
of $C_{in}$ must contain at least one point, and at least one edge of $C_{out}$
must contain at least one point. In both the cases, these five points of $P$ are of different colors.
\end{lemma}

\begin{proof}
{\em Part (i)}: For a contradiction let us assume that three edges of $C_{out}$ contains three 
points, one edge of $C_{in}$ contains a point, and the colors of these 
four points are different satisfying Observation \ref{distinct-color}. Let $e$ be the edge 
of $C_{out}$ containing no point. We start moving the edge $e$ of $C_{out}$ 
containing no point towards the annulus-center $c$ in the self-parallel manner. 
To maintain the same width of $\cal A$ we need to move the similar side of
$C_{in}$ of $e$ towards $c$ simultaneously with $e$ until $e$ hits a point 
of $P$ having color different from the colors of all three points on 
$C_{out}$ (see Fig. \ref{fig:axis_parallel_rect}). Note that, all the points lying 
in the annular region of $\cal A$ remains in the annular region of $\cal A'$ formed 
with the new positions of $e$ (surely, a few more points may enter in the annular region).

If none of the edges of $C_{in}$ contains a point, then we can reduce the width of the annulus by 
moving the four edges of $C_{in}$ away from the annulus-center $c$ in 
self-parallel manner until at least one edge of $C_{in}$ hits a point inside the annular 
region $\cal A$ satisfying Observation \ref{distinct-color}. 

{\em Part (ii)}: Similar proof holds to show that the width of an annulus defined by four 
points on the four edges of $C_{in}$ and one edge of $C_{out}$ containing a 
point. Thus, the lemma follows. 
\qed
\end{proof}
Lemma \ref{four_points} leads to the following result;  
\begin{lemma} \label{width_two_points}
 In an optimum $CSRA$ a pair of similar sides (of $C_{in}$ 
and $C_{out}$) will contain two points of different colors.
\end{lemma}

We now discuss the algorithm for rectangular annulus based on the Lemma 
\ref{width_two_points}. Assume that the points in $P$ are available in
sorted order with respect to $x$- and $y$-coordinates in two arrays
$P_x$ and $P_y$ respectively. Consider a pair of points $p,q \in P$ of different 
colors. We test whether a $CSRA$ is possible with $p$ and $q$ on the top 
boundaries of $C_{out}$ and $C_{in}$ respectively.
The width of such a $CSRA$, if exists, will be $\delta=y(p)-y(q)$. Similar method
is adopted to find the existence of a $CSRA$ with $p,q$ in the bottom, left
or right boundaries of $C_{in}$ and $C_{out}$.

\begin{figure}[h]
 \centerline{\includegraphics[scale=0.4]{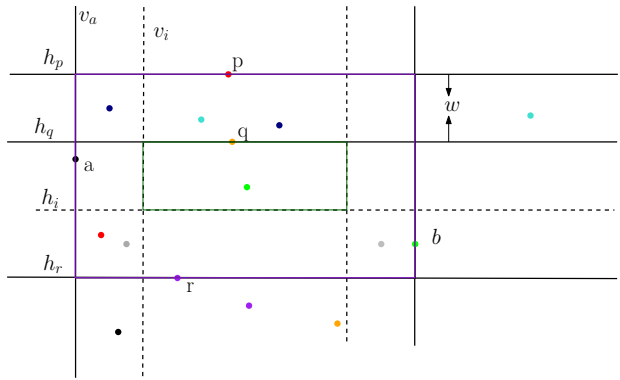}}

\caption{Illustration of rectangular annulus construction} 
  \label{rect_algo}

\end{figure}
Let us consider $h_p$ and $h_q$. In a linear scan we can find a point $s \in 
P_y$ such that the horizontal strip defined by $h_p$ and $h_s$ is color spanning. 
If $y(p)-y(s) \leq 2\delta$, then an annulus of width $\delta$ is trivially
obtained with $C_{in}=\emptyset$. Thus, we assume that $y(p)-y(s) > 2\delta$. 

Observe that, for all points $r \in P$ with $y(r) \leq y(s)$, the horizontal 
strip $H$ defined by $h_p$ and $h_r$ will be color spanning. We now compute 
a minimum width $CSRA$ inside the strip $H$ with points $p$ and $r$ lying 
respectively on the top and bottom boundary of $C_{out}$, and $q$ lying on 
the top boundary of $C_{in}$.

Fix a point $a\in P_y$ with $x(a) < \min(x(p),x(q),x(r))$ inside the strip $H$ 
having color different from that of  $p$, $q$ and $r$. In a linear 
scan in the array $P_x$, we can get a point $b_x$ such that the rectangle 
defined by the lines $h_p,h_r,v_a,v_b$ is color-spanning. Thus for all points $c 
\in P_x$ satisfying $y(c) \in [y(p),y(r)]$ and $x(c) \geq x(b)$, the rectangle 
defined by the lines $h_p,h_r,v_a,v_c$ will be color-spanning. In a line 
sweep inside the strip $H$, we choose those points $c$ with $x(c) \geq x(b)$ 
and having color different from that of $p,q,r,a$. Now, the rectangle $R$ formed 
by $h_p,h_r,v_a,v_c$ defines $C_{out}$. The corresponding $C_{in}$ will be 
co-centric with $C_{out}$, its length and width will be $(x(c)-x(a)) - 2\delta$
(see Fig. \ref{rect_algo}). We can test whether the created annulus is 
color-spanning or not by inspecting the points in $P$ in another linear 
scan\footnote{This can also be tested in poly-logarithmic time by 
maintaining $k$ range trees with points corresponding to $k$ colors separately 
and performing emptiness queries for four axis-parallel rectangles in all 
those range trees; this will increase in the space complexity to $O(n\log n)$.}. 
Thus, we have the following result:
 
\begin{lemma}\label{l}
Given a set of $n$ points, each assigned with one of the $k$ given colors, 
the minimum width $CSRA$ can be computed in $O(n^6)$ time using $O(n)$ space. 
\end{lemma}

\begin{proof}
For each pair of points $p,q \in P$, we execute four loops: (i) choosing the 
points $r$ using a horizontal line sweep, (ii) choosing the points $a$ using a 
vertical line sweep inside the strip $H$, (iii) choosing the points $c$ using a vertical 
line sweep from $a$ towards right, and then (iv) testing whether the created 
annulus is color-spanning by testing the points in the rectangle $R=C_{out}$. 
\qed
\end{proof}
We can improve the time complexity by merging the two loops (iii) and (iv) 
mentioned in the proof of Lemma \ref{l} as follows:

For each $a \in H$, we start sweeping two vertical lines $L_1$ and $L_2$ 
simultaneously by using an array $D$ of size $k$, and a scalar variable $Z$. 
$D[i]$ indicates the number of points of color $i$ in the annulus, and $Z$ 
indicates the number of colors absent in the annulus. We initialize $D[i]=0$ 
for all $i=1,2,\ldots,k$, $D[color(p)]=D[color(q)]=D[color(r)]= D[color(a)]=1$, 
and $Z=k-4$. We also initialize the starting position of $L_2$ as the index 
of the rightmost point $d \in P_x$ with $x(d) < x(a)+\delta$ and $y(d)\in y(p,y(r))$. The sweep of 
$L_1$ is implemented by considering the points in $P_x$ in order from the point 
$a$. For each encountered point $c$ of color $i$ (say), if $y(c) \not\in [y(p),
y(r)]$, then $c$ is not feasible to be within the rectangular annulus currently 
under construction. Otherwise, we do the following: 

\begin{figure}[htbp]

\begin{minipage}[b]{0.45\linewidth}
\centering
\centerline{\includegraphics[scale=0.5]{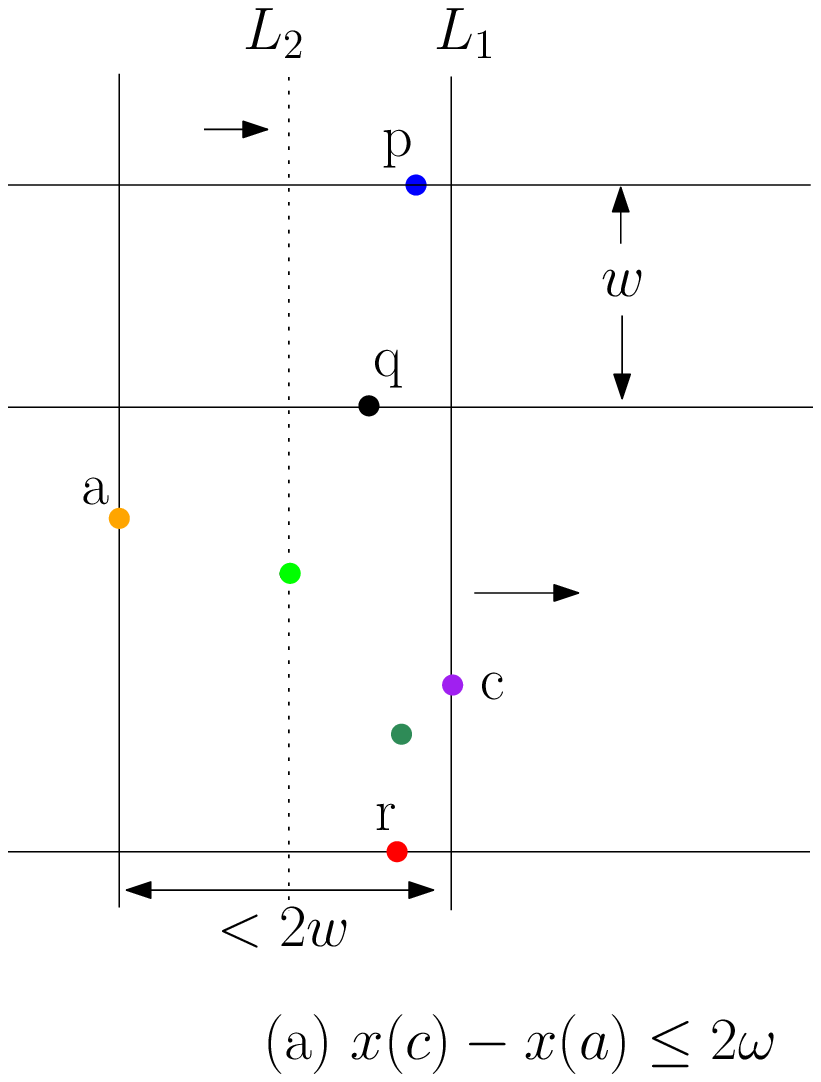}}
\end{minipage} %
\begin{minipage}[b]{0.55\linewidth}
\centering
\centerline{\includegraphics[scale=0.5]{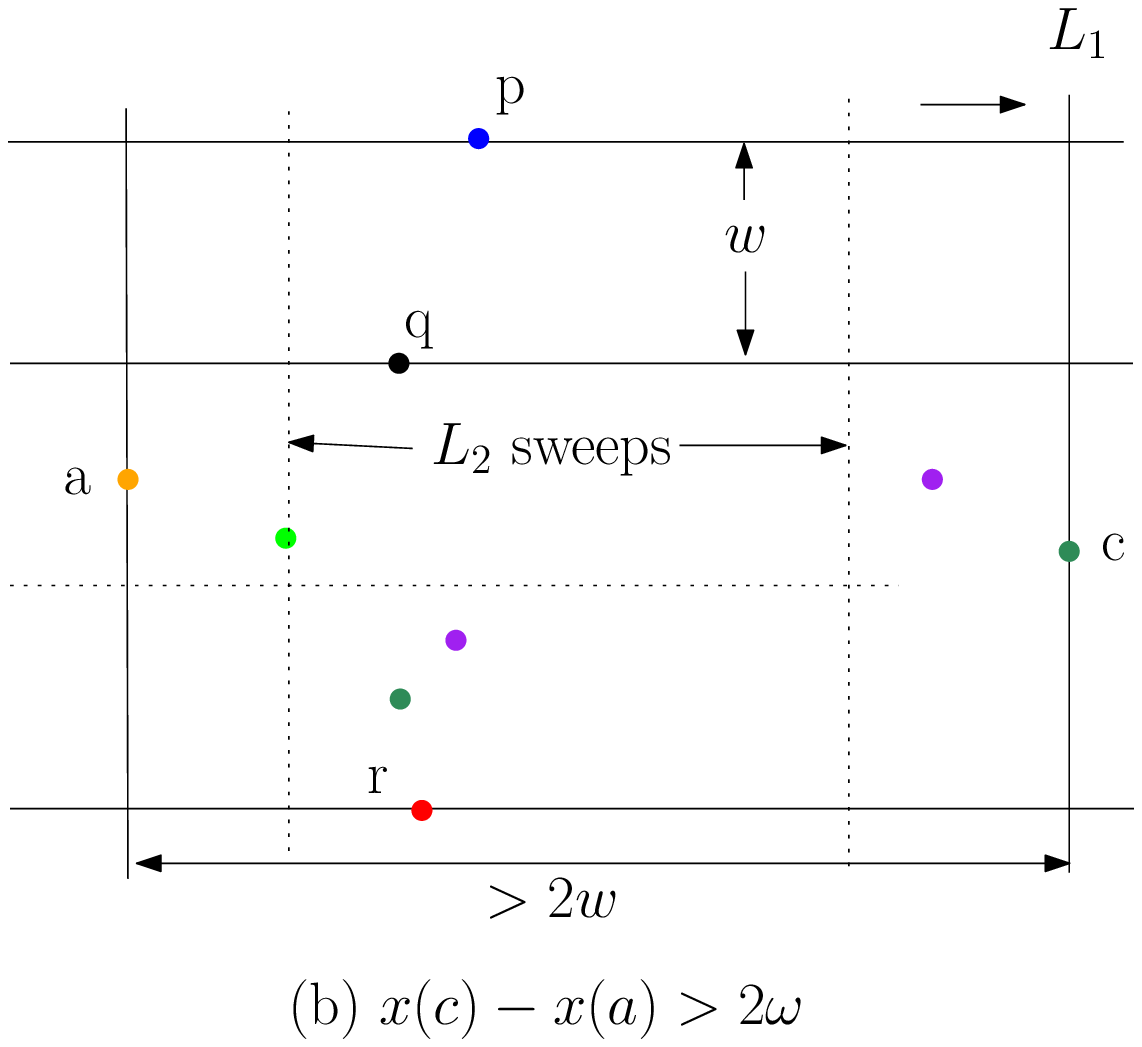}}
\end{minipage} %

\caption{$c$ enters the annulus boundary}
\label{fig:rect_demo}

\end{figure}

\begin{description}
\item[$c$ enters in the annulus: ] we set $D[i]=D[i]+1$. If $D[i]=1$ (a point 
of color $i$ is newly obtained in the annulus), then $Z$ is decremented by 1. 
Now, 
\item[if $x(c)-x(a) \leq 2\delta$] [indicating $C_{in}=\emptyset$ (see Fig. 
\ref{fig:rect_demo}(a))], we need not have to do anything.  
\item[if $x(c)-x(a) > 2\delta$] [indicating $C_{in}\neq \emptyset$ (see Fig. 
\ref{fig:rect_demo}(b))],  
we start sweeping $L_2$ from its present position up to $x(c)-\delta$. For all 
points $d$ encountered by $L_2$, if $y(q)> y(d) > y(r)+\delta$ then $D[j]=D[j]
-1$, where $j=color(d)$. If $D[j]=0$ (indicating no point of color $j$ in the 
annulus) then $Z$ is incremented by 1.  
\item[Check whether the annulus $\cal A$ is color-spanning:] If $Z=0$, then 
report ``success", and stop sweeping of $L_1$.
\end{description} 
Since sweeping of $L_1$ and $L_2$ needs $O(n)$ time, the time complexity of the 
algorithm reduces to $O(n^5)$. Here, it needs to be mentioned that, for each 
point 
$p$, we are testing whether there exists an annulus of width $\delta=y(p)-y(q)$ 
by choosing all possible points $q$ satisfying $y(q) < y(p)$. Thus, for each 
point $p$, we choose $O(n)$ points as $q$ in the worst case.  But, this can be 
improved using a simple binary search technique.
We can find minimum value of $\delta$ (the minimum width of an annulus) with $p$ 
on the top boundary of $C_{out}$ by choosing $q$ using binary search among the 
points in $P_y$ having $y$-coordinate less than $y(p)$ and of different color. 
Thus for each point $p$, we need to choose $O(\log n)$ points as $q$, which 
leads to a total time complexity of $O(n^4\log n)$. Now, this algorithm can also
be implemented in inplace manner using $O(k)$ extra space. For each choice of
$p\in P$ (in decreasing order), we first get $P_y$ in $O(n\log n)$ time (in the 
array $P$ itself) to choose an appropriate member $q\in P$. All the members in 
$P$ below $q$  can serve the role of the point $r\in P$. Now, for $p,q,r \in P$,
we sort $P$ again to get $P_x$ (in the array $P$), and the sweep is performed to
choose $a\in P$ to compute minimum width $CSRA$, if exists. Thus, we need to store
$p,q,r$ (using $O(1)$ variables) to identify the next triple  $p,q',r'\in P$ for processing.
This leads to the following result;

\begin{theorem}
Given a set of $n$ points, each assigned with one of the $k$ given colors, 
the minimum width $CSRA$ can be computed in $O(n^4\log n)$ time using $O(k)$ extra
space. 
\end{theorem}

\section{Color Spanning Equilateral Triangular Annulus} 
\begin{wrapfigure}{r}{0.4\textwidth}
\vspace{-0.37in}
\centerline{\includegraphics[scale=0.4]{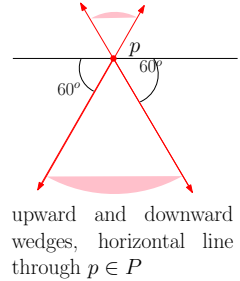}}
 
\caption{Basic constructions}
\vspace{-0.35in} 
\label{basic}
\end{wrapfigure}
A {\em color-spanning equilateral triangular annulus $(CSETA)$} is a 
color-spanning annulus $\cal A$ bounded by two co-centric equilateral triangles
$C_{out}$ and $C_{in}$ where the common center\footnote{the point of intersection of three medians of the equilateral triangle} for both the triangles is 
termed as annulus-center $c$. We assume that the base of $C_{in}$ and $C_{out}$
are parallel to the $x$-axis. Two such types of annulus is possible depending on whether the apex of 
$C_{in}$ and $C_{out}$ is above or below the base of the corresponding triangle. We 
will explain the method assuming that the apex is above the base. The other case can 
be similarly processed. 

%
%
 \begin{figure}[h] 
 \centerline{\includegraphics[scale=0.4]{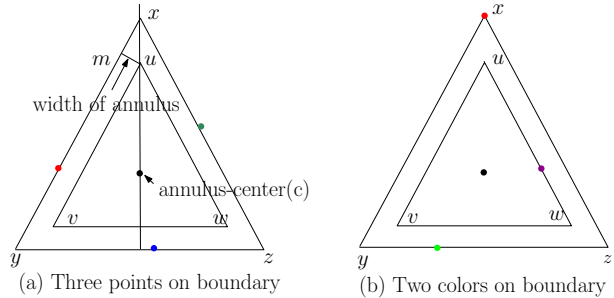}}
 \caption{Two types of boundary points configuration}
  \label{boundary}
\end{figure}
\textbf{\em Wedge:} Consider a pair of half-lines emanating from a point $p$ having 
angles $60^o$ and $120^o$ with the horizontal line at $p$. The point $p$ is said 
to be the vertex of the wedge. The (open) areas created by these two half-lines above and below $p$ are 
termed as the {\em upward wedge} and {\em downward wedge} respectively for the point $p$ (see Fig. \ref{basic}).

\textbf{\em Similar vertex and edge:} Each pair of vertices, one of $C_{in}$ and one of $C_{out}$, that are 
collinear with the annulus-center $c$, are said to be {\em similar vertices}.
Similarly the edges of $C_{in}$ and $C_{out}$ that are 
parallel to each other, are said to be {\em similar edges}.
In Fig. \ref{boundary}(a), $x$ and $u$ are similar vertices, and $\overline{xy}$ 
and $\overline{uv}$ are similar edges.
 
\textbf{\em Width of Annulus:} The width of a triangular annulus is the difference in the 
length of the line segments perpendicular on two similar edges from its 
annulus-center $c$ (see Fig. \ref{boundary}(a)).

\begin{lemma}\label{distance_twice}
The distance among a pair of similar vertices in $C_{in}$ and $C_{out}$ of a 
$CSETA$ $\cal A$ is twice the width of that $CSETA$. 
\end{lemma}

\begin{proof}
 Consider Fig \ref{boundary}(a, where $C_{out}=\triangle xyz$ and 
$C_{in}=\triangle uvw$. The line containing $\overline{xu}$ bisects both $\angle yxz$ and $\angle 
vuw$. The width of $\cal A$ is the perpendicular distance from $u$ to the line 
$\overline{xy}$, which is equals to $|um|$. Thus, $\triangle xum$ is right 
angled with $\angle mxu=30^o$, and 
$|xu|=2\times|um|$.
 \qed
\end{proof}

\begin{lemma} \label{lemma_boundary}
The necessary and sufficient condition for $\cal A$ to be a minimum width 
$CSETA$ is that  (i) all the three edges of $C_{out}$  must 
contain at least one point, and  at least one edge of $C_{in}$ must contain at least one point
or (ii) all the three edges of $C_{in}$  must 
contain at least one point, and  at least one edge of $C_{out}$ must
contain at least one point. In both the cases, these four points of $P$ are of different colors.
\end{lemma} 

\begin{col}\label{col}
 If a point lies on a vertex of $C_{in}$ or $C_{out}$, 
then a $CSETA$ can be defined by three points (instead of four points) on boundary\footnote{it 
can be shown in a similar way as in lemma \ref{lemma_boundary}.} (see Fig. \ref{boundary}(b)).  
\end{col}

%

\subsection{Algorithm}\label{algo1}
Here, we describe the general 
framework of the algorithm to find $CSETA$ of minimum width where 
$C_{out}$ is defined by three points of distinct colors (see Lemma \ref{lemma_boundary}). The
same method works to find a $CSETA$ of minimum width where $C_{in}$ is defined by three points.
We sort the points in $P$ with respect to their $y$-coordinates. Consider each bi-colored pair of 
points $p,q\in P$ with $x(p) < x(q)$. Consider a downward 
wedge $W_{p,q}$ at a vertex  $v_{p,q}$ by drawing a line of angle $60^o$ 
through the point $p$ and a line of angle $120^o$ through the point $q$ (see 
Fig. \ref{construction}). In linear time we can find the points $P_{p,q} 
\subseteq P$ lying in $W_{p,q}$ sorted with respect to their $y$-values. 
If $W_{p,q}$ is not color-spanning, $C_{out}$ cannot be defined by $W_{p,q}$. 
So, let us assume that $W_{p,q}$ is color-spanning and the pair $(p,q)$ defines two boundaries of 
$C_{out}$\footnote{similar idea works when it define $C_{in}$ by considering
the points outside the wedge to define $C_{out}$ instead of considering inner wedge points.}. Let $\ell_{p,q}$ be a 
vertical line through $v_{p,q}$. Now, the following results are important.

 \begin{figure}[h] 
 \centerline{\includegraphics[scale=0.5]{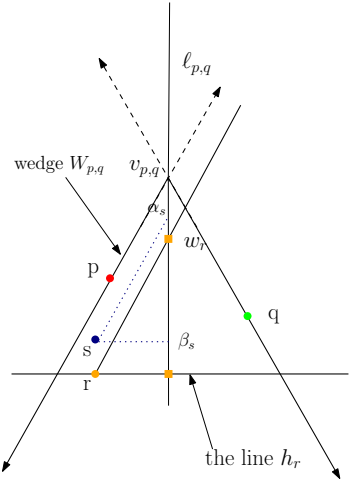}}
 \caption{Construction of wedges and lines}
  \label{construction}
\end{figure}
\begin{lemma} \label{obx} 
The base of $C_{out}$ corresponds to a horizontal line through a point $r 
\in P_{p,q}$, where $y(r) < \min(y(p),y(q))$ $color(r)\not\in \{color(p),
color(q)\}$. 
\end{lemma} 

For each point $r \in P_{p,q}$, we define two lines, namely 

\begin{itemize}
\item[$h_r$:]  the horizontal line through the point $r$, and 
\item[$w_r$:] the wedge line, which is a line of angle $60^o$ (resp. $120^o$) 
through the point $r$ depending on whether $r$ is to the left (resp. 
right) of $\ell_{p,q}$.
\end{itemize}

By Lemma \ref{obx}, we set $r = p$ or $q$ depending on whether $y(p)<y(q)$ or 
$y(p)>y(q)$ (consider the minimum one), and start sweeping a horizontal line. The 
points in $W_{p,q}$ encountered by the sweep line are the event-points. At 
each event-point $r$, if the triangle $\Delta$ defined by the vertex $v_{p,q}$ 
and the line $h_r$ in the wedge $W_{p,q}$ is color-spanning, $r$ defines the 
base of $C_{out}$ and we compute $C_{in}$ as follows:

$\blacktriangleright$ Create an array $D$ of size $k$ to store the distance of the closest 
point of each color $i$ from the boundary of $C_{out}$. Initialize 
$D[i]=\infty$ for all $i=1,2,\ldots, k$. 

$\blacktriangleright$ Let $\theta$ be the point of intersection of $h_r$ and $\ell_{p,q}$.

$\blacktriangleright$ For every point $s \in \Delta$, do the followings:

\begin{itemize}
  \item[$\bullet$]  Let $\alpha_s$ and $\beta_s$ be the points of intersection of $w_s$ 
  and $h_s$ with $\ell_{p,q}$. We will use the term {\em $\alpha$-point} and {\em 
  $\beta$-point} of $s$ to denote the points $\alpha_s$ and $\beta_s$ 
  respectively.  
  \item[$\bullet$] Compute the distances $d(\theta,\beta_s)$ and $d(v_{p,q},
  \alpha_s)$.
  \item[$\bullet$] The distance of $s$ from 
  the boundary of $C_{out}$ is $\mu=\min(d(\theta,\beta_s),\frac{d(v_{p,q},
  \alpha_s)}{2})$ (see Observation \ref{distance_twice}). If $s$ is of color $i$ then store $min(\mu,D[i])$ 
  in $D[i]$.
 \end{itemize}
 
 $\blacktriangleright$ The width of the annulus $\delta=\max_{i=1}^k D[i]$, which can be 
computed by a linear scan in the array $D$. Thus, $C_{in}$ is determined.

\begin{lemma}
The overall time complexity of this simple scheme is $O(n^4)$. 
\end{lemma} 
\begin{proof}
For each of the $O(n^2)$ pair of points $(p,q)$, 
the sweep considers $O(n)$ event points in $P_{p,q}$. For each event point
$r \in P_{p,q}$, we need to spend $O(n)$ time in the worst case to inspect 
all the points in $\Delta$. Thus the result follows. \qed 
\end{proof}
We can improve the time complexity by maintaining two AVL trees $G$ and $H$ 
(instead of the array $D$) while processing each pair of points $(p,q)\in P$. 
Here, $G[i]$ stores the closest point $\alpha$-point of each color $i$ from $v_{p,q}$ and $H$ stores the {\em 
minimum $y$-coordinate} of the $\beta$-points of each color $i$. Each element of $G$ and $H$ is attached with 
the corresponding color index. Auxiliary arrays $G'$ and $H'$ are maintained, 
where $G[i]$ (resp. $H[i]$) stores the position (address of location) containing color $i$ in $G$ 
(resp. $H$). Thus, the size of both $G,H,G'$ and $H'$ are $O(k)$. We have the points 
$P$ in decreasing order of their $y$-coordinates. When a new point $s\in W_{p,q}$ 
of color $i\in \{1,2,\ldots, k\} \setminus\{color(p),color(q)\}$ is faced by the 
sweep line, we process $s$. The position of color $i$ in the array $G$ and $H$ 
are $j_1 = G'[i]$ and $j_2=H'[i]$ respectively. $G[j_1]$ is updated with 
$\min(\frac{d(v_{p,q},\alpha_s)}{2},G[j_1])$, and $H[j_2]$ is updated with $y(s)$. Next, both the AVL trees $G$ and 
$H$ are adjusted in $O(\log k)$ time. Now, define $r=s$ (i.e., the base line at 
point $s$), and scan both $G$ and $H$ in increasing order of their values using a 
merge like pass to find the required $C_{in}$ so that the annular region $\cal A$
is of minimum width and contains every color from either $G$ or $H$. Thus, we
have the following result:

 \begin{theorem}
 Given a set $P$ of $n$ points, each point is assigned with one of the $k$ 
possible colors, the minimum width $CSETA$ can be computed in $O(n^3 k)$ time 
using $O(k)$ additional space. 
 \end{theorem}
\begin{proof}
The correctness of the algorithm follows from Lemma \ref{lemma_boundary}, its corollary, and the algorithm
where we always maintain the nearest points of each color corresponding to the wedge
boundaries and the base line in two $O(k)$ sized AVL trees.
The time complexity is analyzed considering the fact that we need to consider $n\choose 2$ 
pairs of bi-colored points $(p,q) \in P$ to define the wedge $W_{p,q}$. During 
the line sweep for the processing $W_{p,q}$, at each event point $r\in P_{p,q}$ 
we spend $O(k)$ time\footnote{The adjustment of the AVL-trees $G$ and $H$ needs 
$O(\log k)$ time. But, to find the elements from $G$ and $H$ to determine the 
minimum width annulus $\cal A$ needs a linear scan in $G$ and $H$.}.  Hence, 
the time complexity result follows. The extra space requirement follows from the 
size of $G$, $H$, $G'$ and $H'$. 

Corollary \ref{col} suggested that we need to consider $C_{out}$ with each point 
on one of its three vertices. When a point $p\in P$ is considered as apex, the 
wedge is defined by the two lines of angle $60^o$ and $120^o$ with the axis 
through the point $p$ and similar method is executed to get the optimum $CSETA$.
The case when the point $p\in P$ is the left (resp. right) endpoint of the base, 
has already been considered in our algorithm due to Lemma \ref{obx}.
%
Thus the result follows.
\qed
\end{proof}

\end{document}